\documentclass[12pt,a4paper,reqno]{article}

\usepackage{amsmath}
\usepackage{amssymb}
\usepackage{amsthm}
\usepackage{graphicx}
\usepackage[dvips]{color}

\newcommand{\linespace}{\vspace{\baselineskip}}
\newcommand{\semilinespace}{\vspace{0.5\baselineskip}}
\newcommand{\upline}{\vspace{-1.2\abovedisplayskip}}
\renewcommand{\(}{\left(}
\renewcommand{\)}{\right)}

\newcommand{\e}{\mbox{e}}
\renewcommand{\hbar}{\hslash}

\newcommand{\Z}{\mathbb{Z}}

\newcommand{\R}{\mathbb{R}}

\newcommand{\E}{\mathcal{E}}

\renewcommand{\H}{\mathcal{H}}

\newcommand{\T}{\mathcal{T}}
\newcommand{\TF}{\mathcal{T}_\text{F}}

\newcommand{\rmF}{\text{\emph{F}}}

\newcommand{\<}{\langle}
\renewcommand{\>}{\rangle}

\newcommand{\qqquad}{\quad\quad\quad}

\newcommand{\impl}{\Longrightarrow}

\newcommand{\be}{\begin{equation}}
\newcommand{\ee}{\end{equation}}
\newcommand{\Pii}{\widetilde{\Pi}}

\theoremstyle{plain} 
\newtheorem{thm}{Theorem}

\newtheorem{lemma}{Lemma}

\begin{document}

\title{The logic of the future in quantum theory\footnote{Expanded version of a talk given at the Arthur Prior Centenary Conference, Oxford, in September 2014}}


\author{Anthony Sudbery\\[10pt] \small Department of Mathematics,
    University of York,\\[-2pt] \small Heslington, York, England YO10 5DD\\ 
    \small  Email:  as2@york.ac.uk}
\date{27 January 2015, revised 19 March 2016}

\maketitle

\begin{abstract}
According to quantum mechanics, statements about the future made by sentient beings like us are, in general, neither true nor false; they must satisfy a many-valued logic. I propose that the truth value of such a statement should be identified with the probability that the event it describes will occur. After reviewing the history of related ideas in logic, I argue that it gives an understanding of probability which is particularly satisfactory for use in quantum mechanics. I construct a lattice of 
future-tense propositions, with truth values in the interval $[0,1]$, and derive logical properties of these truth values given by the usual quantum-mechanical formula for the probability of a history.  
\end{abstract}
 
\section{Introduction}
\label{intro}

The subject of this paper is the logic governing the expressions of present experience, past memory and future expectations of a subject, regarded as a sentient and articulate physical system and therefore governed by the laws of physics in the framework of quantum theory. Since the expressions to be considered refer to the past, present and future of the subject, they must be governed by some form of temporal logic, such as is found in the work of Arthur Prior and his followers. Quantum theory being an indeterministic theory, the problem of future contingents is particularly relevant. In order to incorporate indeterminism in statements in the future tense, the temporal logic proposed here, differing from Prior's, has truth values which are not bivalent but take real values between 0 and 1 (inclusive) --- the truth of a future-tense statement is identified with the probability that it will turn out to be true.  Maybe this thesis will be of interest to philosophers of probability and philosophers of quantum mechanics as well as temporal logicians.

 Classically, it would be natural to take the past and present as definite: statements in the present and past tenses have truth values 0 or 1. For a apeaker governed by quantum theory, it is not clear that this can be incorporated in the framework suggested here without modification. This is currently being investigated. The present paper, therefore, is restricted to the study of the present and future tenses.


The paper is organised as follows. In section 2 I outline the physical framework being used, that of quantum mechanics, and define the propositions in question: statements about the present and future experiences of the speaker, regarded as a physical system in a world governed by quantum mechanics. This requires a distinction between external statements about the physical world, made from a standpoint outside the world, and internal statements made by a speaker in the world; it is the latter that are discussed here. I argue that the superposition principle of quantum mechanics makes it impossible, in general, for such future-tense propositions to be either definitely true or definitely false, and I introduce my main proposal concerning the relation between probabiliy and the truth values of future-tense propositions in a many-valued logic.  

In the next section I acknowledge some of the problems with this idea, and attempt to trace the idea, and reactions to these problems, in previous writers. In the following section the focus shifts to the other side of the equation, from truth values to probability; I argue that this concept of probability is especially well suited to use in quantum mechanics. In particular, it allows the laws of motion of quantum mechanics (from the internal perspective) to be expressed as an indeterministic theory in continuous time, unlike the conventional postulates on the results of experiments.

 In Section 5 I examine the logical properties of the temporal logic as it might apply to the utterances of a speaker satisfying the laws of quantum mechanics. The logical system consists of a lattice $\T$ of tensed propositions with truth values in the interval $[0,1]$ of real numbers. The lattice $\T$ is constructed from a tenseless lattice $\E$ containing propositions about the experiences of a particular sentient system; the elements of $\T$ are elements of $\E$, each associated with a time, and elements constructed from them by operations $\land, \lor$ and $\lnot$ which are assumed to satisfy the same rules as conjunction, disjunction and negation. A further assumption on the tense operators (that they respect the logical operations in the tenseless lattice $\E$) leads to the characterisation of the elements of the lattice $\T$ as disjunctions of \emph{histories}: conjunctions of statements of experience at different times. 

An assignment of truth values to the elements of $\T$ is defined to be a function from $\T$ to real numbers which takes the value $0$ or $1$ on present and past propositions, but can take any value between $0$ and $1$ for future propositions. The bulk of this section is devoted to exploring the logical properties of the assignments of truth values to future-tense propositions which are obtained from the quantum-mechanical formula for the probability of a single history; the truth value of a disjunction of histories is defined by means of the formula
\[
\tau(p\lor q) = \tau(p) + \tau(q) - \tau(p\land q)
\]
which would be expected for probabililties. It has to be shown that this gives values between $0$ and $1$, and that it is consistent to apply it for all elements $p,q$ of the lattice $\T$. This is the content of Theorem \ref{disj}, to which the preceding chain of theorems has been leading. The usual truth tables for $\land$ and $\lor$ would follow from this if truth values were restricted to $0$ and $1$, but in general the usual relation between the truth value of a compound proposition and the truth values of its components is weakened; the extent of this weakening is laid out in Theorem \ref{truth}. This is followed by a discussion of whether the connective $\lor$ in this logic can legitimately be regarded as a generalised form of ``or" for future-tense statements.

Finally, in Section 6 I take a critical look at the theory adumbrated here and try to assess its relevance to different forms of indeterministic physical theory. I claim that it is the appropriate logic for a classical theory of an open future, and for a quantum theory applied to a world in which a mechanism of decoherence operates, as it does in the world\footnote{Here, and elsewhere, I use the word ``world" in its usual sense, or in the Wittgensteinian sense: ``everything that is the case". No reference is intended to the ``worlds" of the many-worlds interpretation of quantum mechanics.} we live in. The section concludes with a discussion of weaknesses in this proposal and possible alternative formulations. The results of the paper are summarised in the concluding Section 7.

\section{Outline of Quantum Theory}
\label{sec:quantum}

\subsection{Living in the Quantum World}
\label{subsec:qworld}

Quantum mechanics boasts a simple, well-defined mathematical formalism and an unprecedentedly successful record of calculation and prediction in physical applications. Following Schr\"odinger \cite{Schrodinger:interpretns}, Everett \cite{Everett} and Wheeler \cite{Wheeler}, I assume that everything in the mathematical formalism can be taken as a true description of the universe, and that no extra apparatus is needed to relate this description to what we measure and observe.

The clean mathematical formalism, in its nonrelativistic form, offers a single mathematical object $|\Psi(t)\>$ to describe whatever physical system we are interested in --- ultimately the whole universe $U$. This object is a vector in a Hilbert space $\H_U$ (the marks $|\>$ in the symbol indicate the vector nature of $|\Psi\>$), varying with time $t$ in accordance with a first-order differential equation (the Schr\"odinger equation). It follows that if the value of $|\Psi(t)\>$ is known at some one time, its value at all other times is determined.

This paper is concerned with the experiences and utterances of a sentient being $S$ (which could be you or me) regarded as a physical system in the universe. As such, $S$ has states which are described by the elements of a Hilbert space $\H_S$. The states of the rest of the universe (which I will sometimes call ``the outside world") are described by the elements of a Hilbert space $\H_{U'}$, and the state space $\H_U$ of the whole universe is the tensor product
\[
\H_U = \H_S \otimes \H_{U'}.
\]

Among the states\footnote{From now on, I will no longer be careful to distinguish between a state of a physical system and the vector in its Hilbert space which describes that state.}  in $\H_S$ are those in which the sentient being $S$ is having a definite experience which it can describe by present-tense propositions. It is probably true \cite{Tegmark} of the brain states of human beings that these experience states have definite values of various physiological and neurological quantities, and I will take this as part of the definition of a sentient physical system. According to quantum mechanics, then, the states of definite experience are eigenvectors of the corresponding \emph{experience operators}. The full set of eigenvectors of these operators is a basis for the Hilbert space $\H_S$ which I will call the \emph{experience basis}, although it will also include many vectors which do not correspond to experiences (or even to anything which looks like a sentient animal).

I will denote the experiences of $S$ by $\eta_i$, pretending for notational convenience that they are countable, and the corresponding elements of the experience basis by $|\eta_i\>$; abusing notation slightly, I will use the same symbol for a general element of the experience basis. 

Similarly, among the states of the whole universe are tensor products\footnote{In this notation, due to Dirac, juxtaposition of vector symbols $|\>$ denotes the tensor product of vectors.}  $|\Psi\> = |\eta\>|\psi'\>$, in which $|\eta\>\in\H_S$ is an experience state and $|\psi'\>$ is a state of the rest of the universe which conforms to the experience $\eta$ in the sense that if $\eta$ includes the experience of seeing, for example, a tiger about to pounce, then $|\psi'\>$ is a state of the outside world including a tiger about to pounce. Such product states are understandable; this is what we expect the sentient system and the outside world, together, to be like. 

But according to quantum theory, the laws of physics, in the form of the Schr\"odinger equation, will often cause such a product state of the universe to change into a less understandable state
\be\label{entangle}
|\Psi\> = \sum_i |\eta_i\>|\psi'_i\>
\ee
in which a number of different experience states enter, each associated with a different state of the outside world. In such a state the sentient system is \emph{entangled} with the outside world, and has no definite, unique experience. Eq. \eqref{entangle} is the form of the most general state of the universe; since the experiences states $|\eta_i\>$ form a basis for $\H_S$, any vector in $\H_S \otimes \H_{U'}$ can be written in this form. To take account of changes in time, we should write this as
\be\label{universe}
|\Psi(t)\> = \sum_i|\eta_i\>|\psi_i'(t)\>
\ee
in which the experience states $|\eta_i\>$ are independent of time, being a fixed basis of $\H_S$, and the states of the external world, $|\psi'_i\>$, are time-varying vector coefficients with respect to this basis.

There are two ways of approaching this strange equation without doing violence to the laws of physics as encapsulated in the Schr\"odinger equation. Both originate in Hugh Everett's work in the 1950s, and are expressed in the titles of his PhD thesis and his only published paper in physics. The title of his thesis is ``The theory of the universal wave function", which expresses his main motivation: the idea that quantum mechanics, as expressed in \eqref{universe}, should be taken seriously and literally: however hard to understand, this is a statement of fact about the whole universe. (Schr\"odinger also expressed this idea in unpublished lectures in the early 1950s \cite{Schrodinger:interpretns}. The title of Everett's paper is $`\,``$Relative state" formulation of quantum mechanics', which emphasises another way of understanding \eqref{universe}: any statement about the external world made by an observer must be understood as being \emph{relative} to a particular state of that observer: if, for any reason, we are entitled to say that the state of the observer at time $t$ is $|\eta_n\>$, then we are entitled to assert that the state of the rest of the universe is $|\psi'_n(t)\>$. And, of course, as an observer I \emph{am} entitled to say which $|\eta_n\>$ is my experience state, precisely because it is my experience.

These two assertions about the state of the universe appear to be in contradiction. One of them asserts that the truth about the universe at time $t$ is given by the whole state \eqref{universe}; the other that it is given by one of the components $|\eta_i\>|\psi'_i(t)\>$. I now want to argue that they are, nevertheless, compatible, once it is recognised that they belong in different perspectives.

\subsection{Internal and external}
\label{susec:intext}

This contradiction is of the same type as many familiar contradictions between objective and subjective statements. It can be resolved in the way put forward by Thomas Nagel \cite{Nagel:subjobj,Nagel:nowhere}: we must recognise that there are two positions from which we can make statements of fact or value, and statements made in these two contexts are not commensurable. In the \emph{external} context (the God's-eye view, or the ``view from nowhere") we step outside our own particular situation and talk about the whole universe. In the \emph{internal} context (the view from \emph{now here}), we make statements as agents inside the universe. Thus in the external view, $|\Psi(t)\>$ is the whole truth about the universe; the components $|\eta_i\>|\psi'_i(t)\>$ are (unequal) parts of this truth. But in the internal view, from the perspective of some particular experience state $|\eta_i\>$, the component $|\eta_i\>|\psi'_i(t)\>$ is the actual truth. I may know what the other components $|\eta_j\>|\psi'_j(t)\>$ are, because I can calculate $|\Psi(t)\>$ from the Schr\"odinger equation; but these other components, for me, represent things that \emph{might} have happened but \emph{didn't}. From the external perspective the universal state vector $|\Psi(t)\>$ is a true description of the world at time $t$, and each component $|\eta_j\>|\psi'_j(t)\>$ is just part of that description; from the internal perspective of the experience state $|\eta_i\>$ at time $t$, the state vector $|\eta_i\>|\psi'_i(t)\>$ is a true and complete description of the world, and the universal state vector $|\Psi(t)\>$ is not a description of the world but a true statement of how the world might change.

\subsection{Time and chance}
\label{subsec:tchance}

Two questions present themselves. First, time: if it is unproblematic that at time $t$ I can identify my experience $|\eta_i\>$ and the corresponding state $|\psi'_i(t)\>$ of the rest of the universe, what can I say about other times? For past times, memory might be expected to provide an answer: it seems that we are constructed in such a way that each experience state $|\eta_n\>$ that actually occurs at time $t$ (i.e. for which $|\psi'_n(t)\>$ is non-zero) contains information about a unique experience state at each time before $t$. However, quantum mechanics apparently does not allow such a simple model of memory, for reasons related to the ``watched-pot effect" \cite{MisraSuda, obsdecay}. For this reason I will only treat present and future tenses in this paper, and leave the past tense for future work.

 In general, for a being having a definite experience $\eta_i$ at time $t$, there is nothing in the physics to pick out the experience that the being will have at some future time $t'$: there is no ``thin red line" stretching into the future; no fact of the matter, even relative to the experience $\eta$, about what the subject's experience will be at time $t'$.\footnote{It might seem hard to abandon the assumption, by which we seem to live our lives, that there is something that is going to happen tomorrow, whether or not it is determined by the laws of physics; this might look like just another unacceptable example of scientists dismissing as an illusion a perception which is too immediate to be deniable (like consciousness or free will). However, it has been argued (\cite{openfuture}, \cite{Wallace:multiverse} p. 274) that the openness of the future accords better with our pre-scientific intuition than the existence of a definite future.} The Schr\"odinger equation can be applied to $|\eta_i\>|\psi'_i(t)\>$ to yield a state vector $\e^{-iH(t'-t)/\hbar}\big(|\eta_i\>|\psi'_i(t)\>\big)$ to which it will evolve at a given future time $t'$; but in general this will not be another product state in which the subject has a definite experience, but a superposition of such states in which the subject is entangled with the outside world. 

Second, probability: quantum mechanics is an indeterministic theory, describing chance events. Its empirical success rests on its ability to give probabilities for such events. But how can there be any place for probability in the framework I have described? In the external view, $|\Psi(t)\>$ develops deterministically according to the Schr\"odinger equation; nothing is left to chance. In the internal view, as we have just seen, there are no definite future events and therefore, it seems, no chances. How can ``the probability that my experience will be $\eta_j$ tomorrow" mean anything if ``my experience will be $\eta_j$ tomorrow" has no meaning?

\subsection{Probability and truth}
\label{subsec:probtruth}
 
 My answer to these two questions is to propose that they answer each other. If there is no experience that can be identified as the experience I will have at a future time $t'$, then for each $\eta$ the statement ``I will experience $\eta$ at $t'$" is not true. Nevertheless, if that experience is very likely then this statement is nearly true; in other words, it has a \emph{degree of truth} less than, but close to, 1. This leads to the suggestion that the probability of a future event $E$ should be identified as the degree of truth of the future-tense statement ``$E$ will occur". 
 
This argument reflects the way that probabilities are calculated in quantum mechanics.  The vector $\e^{-iH(t'-t)/\hbar}\big(|\eta_i\>|\psi'_i(t)\>\big)$ to which an experience eigenvector $|\eta_i\>|\psi'_i(t)\>$ evolves from time $t$ to time $t'$, following the Schr\"odinger equation, will not in general be an experience eigenvector; but it may be near an experience eigenvector, and the nearness is measured by a number (close to 1 if the vectors are near to each other) which, in conventional quantum mechanics, is identified (according to the ``Born rule") with a probability .
 
 What exactly is a ``degree of truth" for a future-tense statement? In the rest of the paper I will explore the possibility that it can be taken to be a truth value in a temporal logic.

\section{Truth values as probabilities}
\label{probtruth}

\subsection{Problems}
\label{subsec:probprobs}

{\bf 1.} One function of truth values is to define the meaning of logical connectives such as \emph{and} and \emph{or} by means of truth tables. The existence of such tables is made a prerequisite for many-valued logics in Gottwald's comprehensive treatise \cite{Gottwald:manyvalued}. This is not possible with probabilities. Using $\tau$ to denote the probability of a proposition, the values of $\tau(p\land q)$ and $\tau(p\lor q)$ are not determined as functions of $\tau(p)$ and $\tau(q)$; the two equations which this would require are replaced by a single relation
\be\label{eq:prob}
\tau(p\land q) + \tau(p\lor q) = \tau(p) + \tau(q)
\ee
with inequalities
\begin{align}\label{ineq:prob}
0 \le \tau(p\land q) &\le \min\{\tau(p),\tau(q)\},\\
\max\{\tau(p),\tau(q)\} &\le \tau(p\lor q) \le 1.
\end{align}

{\bf 2}. Another use of truth values is to identify logical tautologies, which are defined as formulae, with propositional variables, which have truth value 1 for all assignments of truth values to the variables. If truth values were replaced in this definition by probabilities satisfying \eqref{eq:prob}, it would not yield a useful set of tautologies. Instead, the usual procedure in probability logic \cite{Adams:problogic} is to require that the value of probability should be 1 for all logical tautologies, assuming that these have already been identified.

\subsection{History}
\label{subsec:history}

Aristotle's discussion of future contingents in \emph{De Interpretatione} is often taken to support a three-valued logic for statements about the future. But Aristotle says more than denying truth or falsity to ``There will be a sea-battle tomorrow"; he also notes that it may be more or less \emph{likely} that there will be a sea-battle. Thus if there is a case for regarding Aristotle as a proponent of many-valued logic for future-tense statements, there might also be a case that he would regard the appropriate truth values as probabilities.

\L ukasiewicz's first system of many-valued logic \cite{Luk:prob} had truth values related to probabilities, though of a rather different kind from those considered here (he was concerned with propositions containing variables, giving them truth values equal to the proportion of values of the variable which made the proposition true). Later \cite{Luk:indeterminism} he was motivated by the problem of future contingents and expressed a preference for many-valued logic in which the truth values could be any real number between 0 and 1; he asked how this was related to probability theory. In his system the truth values do not satisfy \eqref{eq:prob}, but they satisfy \eqref{ineq:prob} with inequality replaced by equality, so that the logical connectives are truth-functional. 

The failure of probabilities to be truth-functional led to a reluctance among logicians to accept them as truth values, but this view was defended by Reichenbach \cite{Reichenbach:probty} (though his frequentist conception of probability was different from that espoused here) and especially by Rescher \cite {Rescher:manyval}. Both these authors claimed to prove that the tautologies of ordinary two-valued logic can all be obtained as tautologies of probability logic, but Reichenbach's argument depended on his frequentism and Rescher's, though more formal, required axioms which themselves refer to the classical tautologies, giving his argument an element of circularity.

In a discussion of time, quantum mechanics and probability, Saunders \cite{Saunders:Synthese3} puts forward a view of time and probability similar to that of this paper: ``events in the future \ldots are \emph{indeterminate}; \ldots probabilities \ldots express the degrees of this indeterminacy" (emphasis in the original). The identification of probabilities with truth values of future-tense propositions has been developed  by Pykacz (\cite{Pykacz} and references therein), who applies this idea to quantum mechanics and uses it to throw interesting light on the GHZ paradox. Pykacz's many-valued logic, however, is different from the one developed here. 

The notion of degrees of truth also occurs in fuzzy logic \cite{Zadeh:fuzzy,Edgington,Sainsbury}, and in this context also Edgington proposed replacing the classical truth tables by the relation \eqref{eq:prob}. More complicated truth values occur in topos theory \cite{Lawvere}, which has been proposed as a suitable logical formalism for the foundations of physics by Isham  \cite{Isham:topos} and for the discussion of partial truth by Butterfield \cite{Butterfield:partial}. D\"oring and Isham have identified probabilities with truth values in topos theory \cite{DoringIsham:prob}, but the probabilities they discuss seem to be credences rather than chances, occurring in both classical and quantum physics but, in the latter, relevant to mixed states and not to pure states. Isham has also formulated a temporal form of quantum logic \cite{Isham:temporalqlogic} in which histories are treated as propositions.

\section{Probabilities as truth values}

In the previous section the focus was on the truth values of future-tense propositions; it was argued that in a world governed by quantum mechanics, such truth values should be equated with probabilities. In this section I will examine the other side of this equation, and argue that the most appropriate understanding of probability in quantum mechanics is that it is a truth value of a future-tense proposition --- or at least that such an understanding is as good as any other.

There are several different concepts that go under the name of ``probability" (or, if you think that it is a single concept, several theories of probability) (\cite{Gillies}, \cite{Wallace:multiverse} part II). Probability can be a frequency, or proportion, of favourable instances among all relevant instances; the probability of a proposition can be a logical property (``degree of entailment") of a proposition relative to another proposition; the probability of an event can be a subjective concept, the degree of belief of a particular person, measured by the betting odds on the occurrence of the event that are acceptable to that person; relatedly, it can be a guide to action for a person; or it can be an objective property, the \emph{propensity} of a physical set-up to produce the event.  

All of these have been put forward as appropriate ways to understand the concept of probability in physics. Frequentist theories, with their apparent no-nonsense relation to empirical data, are attractive to physicists, but it is difficult to make them coherent and non-circular \cite{Wallace:multiverse}.  Classical deterministic physics has a place for probability in situations of incomplete knowledge, which is similar to logical probabiilty; some theories, following the lead of David Bohm \cite{Bohmian}, hold that probability in quantum mechanics arises in the same way. However, such theories, known as ``hidden variable" theories, are generally regarded as having no empirical support. 

The subjective concept of ``degree of belief" is held to provide the meaning of all statements of quantum mechanics by one school of physicists, the quantum Bayesians or ``QBists" \cite{QBism}, but even the apparently opposite camp of Everettians, who believe in the objective truth of a wave function of the universe, appeal to an agent's degree of belief, manifested in their acceptance of betting odds, to give meaning to probabiliity. Thus Wallace (\cite{Wallace:multiverse} p. 122) treats probability as ``\emph{definitionally} linked to rational decision-making" (my emphasis). 

Many, however, would reject this subjectivist approach as abandoning the scientific ideal of an objective theory. It might seem that this is the appropriate concept in the context of this paper, based as it is on the experience of an individual subject and their internal statements about the world. However, ``internal" is not the same as ``subjective"; the sentient being of this paper can aspire to statements which are objectively true in the experience state in which they find themselves. The criterion for success in the quest for objectivity, in this internal context, is the same as in any other view of the world: I am reassured that my statement is objective if others with whom I am in communication agree with me.

Apart from the Bohmians, all schools of thought regard probability in quantum mechanics as dealing with the outcomes of human actions. Textbooks, following Bohr and Heisenberg in what is known as the ``Copenhagen interpretation", still give probabilities only for the results of experiments or for the values of physical quantities found in a measurement; Wallace, for Everettians, considers the decisions of an \emph{agent} and bases his proof of the Born rule on the choices of this agent between different actions; Fuchs et al., for QBists, consider probabilities to be our beliefs about the response of the world to our interventions in it.

But consider my situation. I have retired, and I am very lazy. I spend my days sitting in my garden reading newspapers and watching my grandchildren grow up. I don't do anything, but I'm interested in what's happening, what has happened and what is going to happen. I know what I experience now; I have an excellent memory, and I know what I experienced at all times in the past, or at least at a cloely-packed series of times in an open interval extending up to the present. From these experiences I infer some of what has happened elsewhere in the wider world. But I don't know what will happen in the future. I need some kind of theory to tell me, but our best theory, quantum mechanics, offers me no certainty. OK, what is \emph{likely} to happen in the future? This does not mean ``What do I happen to believe about the future?" I have no beliefs about the future, I simply have no idea what is likely to happen; I depend on physics to tell me.

I do no experiments, make no measurements; according to the Copenhagen interpretation, quantum mechanics will give me no probabilities. I take no action, make no decisions; decision theory is irrelevant to me. I am not a gambler, I never place a bet; I am not interested in the guide to action that Wallace offers me as probabilities, and anyway, because I take no action I do not satisfy the axioms from which he derives those probabilities. Does this mean that probability is meaningless for me? No: I am curious about the future, I want to know what could possibly happen and how likely the different possibilities are. Moreover, I want to know about all times in my future, not just the times when someone somewhere takes it into their head to make a measurement.

What could ``probability" mean for a passive observer like me? That is the question to which this paper proposes an answer: there is such a thing as objective probability, or chance; it applies only to future events, and it consists of the truth value of the statement that such an event will occur.

This is not just a matter of the idle curiosity of old codgers sitting in their gardens. Passive observation, as opposed to active experimentation, is an important element in scientific practice, giving quantitative empirical data which need theoretical calculation --- for example, the time of decay of an unstable system. Such calculations are performed, of course, but they do not have a secure basis, firmly grounded in a well-defined theory \cite{verdammte}; the theoretical description of a measurement of time of decay \cite{obsdecay} has a somewhat uncertain connection to the textbook axioms governing the results of experiments in quantum mechanics. A theory of probabilities as truth values, which can be applied to a continuous range of times, allows a tighter treatment of decay problems from first principles. This will be elaborated elsewhere.

Although quantum theory introduced indeterminism at a fundamental level in physical science, the form of the theory obscures this. The only clear, well-defined law of time evolution in conventional quantum mechanics is the Schr\"odinger equation, which is deterministic: it is a mathematical procedure which takes as input a mathematical description of the state of a physical system at time $t_0$ and produces as output a mathematical description of a single state of the system at a later time $t$. One would expect an indeterministic theory to take the same input and produce as output a set of possible states at the later time $t$ and a probability distribution over this set. Text-book quantum mechanics does not do this; but it is exactly what is provided in the internal perspective described here.

\section{The logic of tensed propositions in our \\quantum world}
\label{sec:logic}

\subsection{The lattice of propositions}
\label{subsec:lattice}

Living in the quantum world, as we do and as we are described in section \ref{subsec:qworld}, what can we say about it? Statements about our possible experiences form a Boolean lattice related to the experience state vectors $|\eta_i\>$. In the fiction that there is a countable basis $|\eta_i\>$, reports of the experiences $\eta_i$ are atoms in this lattice; more generally, it is a Boolean sublattice $\E$ of the lattice of closed subspaces of the Hilbert space $\H_S$.  To form the lattice $\T$ of statements that we want to make about our experience we need maps $N:\E\to\T$ (to give statements about our present experience), $P_t:\E\to\T$ for each positive real number $t$ (to give statements about our experience a time $t$ in the past) and $F_t:\E\to\T$ (for the future). The complete lattice $\T$ is then generated by the images $N(\E)$, $P_t(\E)$ and $F_t(\E)$. This lattice consists of statements that we make, and should conform to the structure of our language; I therefore assume that the lattice has a classical structure, and in particular that $\land$ distributes over $\lor$. (The non-classical features will enter when we consider the relation of this lattice to the physical world, given by truth values --- ``everything that is the case"). The map $N$ simply embeds the lattice $\E$ into $\T$ --- it adds the word ``now" to a report of an experience --- and therefore respects the structure of these reports, so $N$ is an injective lattice homomorphism. Since we are assuming that memory gives reports of past experience with the same quality of definite truth or falsehood as present experiences, I also take the maps $P_t$ to be injective homomorphisms. 

It is not so clear that the future operators $F_t$ should be homomorphisms, but it is not clear what ``and" and ``or" should mean for tensed statements with truth values between $0$ and $1$. I will assume that each $F_t$ is a homomorphism in order to give a meaning to $p\land q$ and $p\lor q$ for all tensed statements $p$ and $q$, but these meanings might be different from ``$p$ and $q$" and ``$p$ or $q$". This is discussed further at the end of this section. A major objective of the mathematical development in this section is to delineate the differences between $\land$ and ``and", and between $\lor$ and ``or"; this is the content of Theorem \ref{truth}.

On the assumption that $\T$ is a distributive lattice, it follows that every proposition in $\T$ is a disjunction of \emph{histories} $h_\text{P}\land h_0 \land h_\text{F}$ where
\be\label{history}
h_0 = N(\Pi_0), \quad h_\text{F} = F_{t_1}(\Pi_1)\land \ldots \land F_{t_n}(\Pi_n) \; \text{ with } \; 0 < t_1 < \cdots < t_n,
\ee
and $h_\text{P}$ is formed similarly with past operators $P_t$. Here each $\Pi_k$ represents an element of the lattice of subspaces $\E$, being the linear operator of orthogonal projection onto the subspace. I will refer to $n$ as the \emph{length} of the history $h_\text{F}$. Since the departures from classical logic in this system occur only in future-tense propositions, the rest of this section will be concerned only with the sublattice $\TF$ generated by $F_t(\E)$ for all $t$.

\subsection{The truth of histories: Conjunction}
\label{subsec:truth}

Truth values are assigned to elements of $\T$ from the perspective of a particular experience $\eta_0$ at a time $t_0$. Past and present propositions are taken to obey classical logic, so any element $N(\Pi)$ or $P_t(\Pi)$ has a truth value of $0$ or $1$, and elements of the sublattice generated by these have truth values determined by the usual truth tables. 

The truth value of a future proposition $F_t(\Pi)$, however, is equated with its probability and could lie anywhere in the closed unit interval $[0,1]$. It is determined by quantum mechanics as follows. The component of the universal state vector determined by the experience $|\eta_0\>$ at the time $t_0$ is $|E_0\> = |\eta_0\>|\psi'_0(t_0)\>$, which would evolve by the Schr\"odinger equation to $\e^{-iHt/\hbar}|E_0\>$ after a time interval $t$, where $H$ is the universal Hamiltonian. On the other hand, the experience state $|\eta_j\>$ will, after the lapse of time $t$, be associated with the component
\[
|\eta_j\>|\psi'_j(t_0 + t)\> = (\Pi_j\otimes 1)|\Psi(t_0 + t)\> = (\Pi_j\otimes 1)\e^{-iHt/\hbar}|E_0\>
\]
of the universal state vector. The geometrical measure of the closeness of these two vectors is taken to be the truth value of the statement ``I will have experience $\eta_j$ after a time $t$" in the context of experience $\eta_0$ at time $t_0$ (from now on this context will be understood):
\be\label{eq:onetimetau}
\tau\left(F_t(\Pi_j)\right) = \left|\<E_0|\Pii_j|E_0\>\right|^2
\ee
where $\tau$ denotes truth value and
\[
\Pii_j = \e^{iHt/\hbar}(\Pi_j\otimes 1)\e^{-iHt/\hbar}.
\]
Eq. \eqref{eq:onetimetau} is the usual expression (the Born rule) for the probability in quantum mechanics.

To extend this to a conjunction of future-tense propositions, i.e.\ to a history $h_\text{F}$ given by \eqref{history}, we adopt the standard extension of the Born rule \cite{Griffiths:book,Wallace:multiverse} to the probability of a history $h_\text{F} = F_{t_1}(\Pi_1)\land \ldots \land F_{t_n}(\Pi_n)$:
\begin{align}\label{tau}
\tau(h_\text{F}) &= \<E_0|\Pii_1\ldots\Pii_{n-1}\Pii_n\Pii_{n-1}\ldots\Pii_1|E_0\>\\
&= \<E_0|C_h C_h^\dagger|E_0\>\notag
\end{align}
where $C_h$ is the \emph{history operator}
\be\label{historyop}
C_h = \Pii_1\cdots\Pii_n.
\ee
Note that if $t_1 = t_2$,
\[
\tau(h_1\land h_2) = \<E_0|\Pii_1\Pii_2|E_0\> = \<E_0|\Pii_1^2\Pii_2|E_0\> = \<E_0|\Pii_1\Pii_2\Pii_1|E_0\>
\]
since the projectors $\Pi_1$ and $\Pi_2$ commute and therefore so do $\Pii_1$ and $\Pii_2$ if $t_1 = t_2$. So the formula \eqref{tau} for $\tau(h_1\land h_2)$ holds for $t_1 = t_2$ as well as $t_1 < t_2$.

I will now explore the logical properties of this definition. First we note the elementary fact
\begin{lemma}
Let $\Pi$ be a projection operator and $|\psi\>$ any state vector. Then

\linespace \emph{(i)}\upline
\[
0 \le \<\psi|\Pi |\psi\> \le \<\psi|\psi\>,
\]

\emph{(ii)}\upline
\[ 
\<\psi|\Pi |\psi\> = 1 \; \iff \; \Pi|\psi\> = |\psi\>.
\]
\end{lemma}
\begin{proof}
\[
\<\psi|\psi\> - \<\psi|\Pi|\psi\> = \<\psi|(1 - \Pi)|\psi\> = \<\psi|(1-\Pi)^2|\psi\> \ge 0,
\]
with equality if and only if $\Pi|\psi\> = |\psi\>$.
\end{proof}

\begin{thm}
For any future history $h_\text{\rmF}$,
\[
0\le \tau(h_\text{\rmF}) \le 1.
\]
\end{thm}
\begin{proof}
By repeated application of Lemma 1(i), using $\<E_0|E_0\> = 1$. 
\end{proof}

\begin{thm} For any two future histories $h_1, h_2$,
\[
\tau(h_1\land h_2) = 1 \;\iff\; \tau(h_1) = \tau(h_2) = 1.
\]
\end{thm}
\begin{proof}
First note that if $k_1,\ldots k_n$ are one-time histories $k_n = F_{t_n}(\Pi_n)$ with $t_1\le \cdots \le t_n$,
\[
\tau(k_1\land \ldots \land k_n) = 1 \;\iff\; \tau(k_1) = \cdots = \tau(k_n).
\]
This is proved by induction on $n$, using Lemma 1(ii). Now if $h_1$ and $h_2$ are any two future histories, $\tau(h_1\land h_2)$, $\tau(h_1)$ and $\tau(h_2)$ are all conjunctions of one-time histories, so both sides of the equivalence in the theorem are equivalent to $\tau(k) = 1$ for all one-time histories occurring in $h_1$ and $h_2$.
\end{proof}
{\bf Corollary}
$\quad\tau(h_1\land h_2) \neq 1 \;\impl\; \tau(h_1)\neq 1\; \text{ or }\; \tau(h_2) \neq 1$.

\begin{thm} For one-time histories $h_1,h_2$ with $t_1 < t_2$,
\[
\tau(h_1) = 0 \; \impl\; \tau(h_1\land h_2) = 0.
\]
\end{thm}
\begin{proof}
\[
\tau(h_1) = 0 \;\impl\; \Pii_1|E_0\> = 0 \;\impl\; \tau(h_1\land h_2) = 0.
\]
\end{proof}

\begin{thm} For one-time histories $h_1,h_2$ with $t_1<t_2$,
\[
\tau(h_1\land h_2) + \tau(h_1\land\lnot h_2) = \tau(h_1)
\]
\end{thm}
\begin{proof}
\begin{align*}
\tau(h_1\land h_2) + \tau(h_1\land\lnot h_2) &= \<E_0|\Pii_1\Pii_2\Pii_1|E_0\> +  \<E_0|\Pii_1\(1 - \Pii_2\)\Pii_1|E_0\>\\
&= \<E_0|\Pii_1|E_0\> = \tau(h_1).
\end{align*}
\end{proof}
{\bf Corollary} $\tau(h_1\land h_2) \leq \tau(h_1)$.

\linespace

Theorems 2, 3 and 4 show that the truth values of two-time histories (conjunctions of one-time propositions) have some of the properties that we would expect for the truth and falsity of conjunction. However, the strict falsity of a conjunction (as opposed to its lack of truth) does not imply the strict falsity of one of the conjuncts. This is not surprising, given the probabilistic nature of the truth values. The quantum nature of the truth values becomes apparent in Theorems 3 and 4
which display a lack of symmetry between the conjuncts. The falsity of a proposition $h_2$ at the later time $t_2$ does not imply the falsity of the conjunction $h_1\land h_2$, because the interposition of a fact $h_1$ at $t_1$ affects the truth of $h_2$ at $t_2$. However, this quantum effect should not be visible at the level of our experience. In order to restore the symmetry of conjunction, to make it possible to extend some theorems which would otherwise only apply to one-time histories, and to complete the logic generally, we need the following \emph{Consistent Histories} assumption:

{\bf CH}\hspace{0.5em} Let $h = F_{t_1}(\Pi_1)\land \cdots \land F_{t_n}(\Pi_n)$ be any history in the lattice of temporal propositions, and for any binary sequence $\alpha = (\alpha_1,\ldots ,\alpha_n)$ ($\alpha_n = 0$ or $1$), let
\[
h_\alpha = F_{t_1}(\Pi_1^{\alpha_1})\land \ldots \land F_{t_n}(\Pi_n^{\alpha_n}).
\]
where $\Pi^0 = \Pi$, $\Pi^1 = \lnot\Pi = 1 - \Pi$. Then
\[
\<E_0|C_{h_\alpha} C_{h_\beta}^\dagger|E_0\> = 0 \;\text{ if }\; \alpha \neq \beta.
\]
where $C_h$ is the history operator of \eqref{historyop}.

This is part of a much wider assumption that demarcates the admissible histories in the ``consistent histories" formulation of quantum mechanics \cite{Griffiths:book}, and can be justified for macroscopic states like our experience states by decoherence theory \cite{Wallace:multiverse}. We do not need the full strength of the consistent histories formulation.

\begin{lemma}\label{PiiPii} If CH holds, the truth value of a history $h = F_{t_1}(\Pi_1)\land\ldots\land F_{t_n}(\Pi_n)$ is given by
\be\label{eqPiiPii}
\tau(h) = \<E_0|\Pii_1\ldots \Pii_n|E_0\>.
\ee
\end{lemma}
\begin{proof} By CH, for each non-zero $(\alpha_1,\ldots , \alpha_n) \in \{0,1\}^n$ we have
\be\label{CH}
\<E_0|\Pii_1\cdots\Pii_n\Pii_n^{\alpha_n}\cdots\Pii_1^{\alpha_1}|E_0\> = 0.
\ee
Taking $\alpha_n = \delta_{ir}$ for some $r$ gives
\be\label{onetimemissing}
\tau(h) = \<E_0|\Pii_1\cdots\Pii_n\Pii_n\cdots ]\Pii_r[\cdots \cdots \Pii_n|E_0\>
\ee
where $]\Pii_r[$ denotes that $\Pii_r$ is omitted from the product. We now prove by downward induction on $r$ that for any subset $R = \{i_1,\ldots ,i_r\}$ of $\{1,\ldots ,n\}$,
\be\label{induction}
\tau(h) = \<E_0|\Pii_1\cdots\Pii_n\Pii_{i_r}\ldots\Pii_{i_1}|E_0\>.
\ee
Taking $\alpha$ so that $\alpha_n = 0$ if $i\in R$, $\alpha_n = 1$ if $i\notin R$, \eqref{CH} gives the right-hand side of \eqref{induction} as a sum of terms with $r+k$ factors to the right of $\Pii_n$, with $\binom{n-r}{k}$ terms if $k>0$, all multiplied by $(-1)^{k+1}$, and all equal to $\tau(h)$ by the inductive hypothesis. This sum is
\[
\sum_{k\neq0}(-1)^{k+1}\binom{n-r}{k}\tau(h) = \tau(h).
\]
 \end{proof}  

We can now restore symmetry to Theorems 3 and 4 and extend them to multi-time histories.

\begin{thm} If CH holds,
\[
\tau(h_1) = 0  \; \impl \; \tau(h_1\land h_2) = 0
\]
for any two histories $h_1$ and $h_2$.
\end{thm}
\begin{proof}
First take $h_2 = F_{t_r}(\Pi_r)$ to be a one-time history and $h_1$ to be given by 
\[
h_1 = F_{t_1}(\Pi_1)\land\cdots\land F_{t_{r-1}}(\Pi_{r-1})\land F_{t_{r+1}}(\Pi_{r+1})\land\cdots \land F_{t_n}(\Pi_n).
\]
If CH holds, $\tau(h_1\land h_2)$ is given by \eqref{onetimemissing}. But if $\tau(h_1) = 0$,
\[
\Pii_n\cdots ]\Pii_r[ \cdots \Pii_1|E_0\> = 0
\]
so that $\tau(h_1\land h_2) = 0$.

Now any $h_2$ can be written as a conjunction of one-time histories, say $h_2 = h_{21}\land\cdots\land h_{2k}$, so
\begin{align*}
\tau(h_1) = 0\; \impl\; \tau(h_1\land h_{21}) = 0\; \impl \cdots &\impl \tau(h_1\land h_{21}\land\cdots\land h_{2k}) = 0\\
 &\impl \; \tau(h_1\land h_2) = 0.
\end{align*}
\end{proof}

\begin{thm}\label{hnoth} If CH holds, and $h_2$ is a one-time history,
\[
\tau(h_1\land h_2) + \tau(h_1\land\lnot h_2) = \tau (h_1).
\]
\end{thm}
Note that in general $\lnot h_1$ is a disjunction of histories, so $\tau(\lnot h_1\land h_2$) is not yet defined.
\begin{proof}
Suppose that $h_1$ and $h_2$ are as in the proof of Theorem 5, so that 
$\lnot h_2 = F_{t_r}(1 - \Pi_r)$. Then the result follows immediately from Lemma \ref{PiiPii}.
\end{proof}

\begin{thm} If CH holds,
\[
\tau(h_1\land h_2) \le \tau(h_1)
\]
for any two future histories $h_1, h_2$.
\end{thm}
\begin{proof}
By Theorem 6, the inequality holds if $h_2$ is a one-time history. Now we argue by induction on the length of $h_2$, using the associativity of $\land$. If $h_3$ is a one-time history and the inequality holds for $h_1$ and $h_2$,
\[
\tau(h_1\land h_2\land h_3) \le \tau(h_1\land h_2) \le \tau(h_1)
\]
so the inequality holds for $h_1$ and $h_2\land h_3$.
\end{proof} 

\begin{thm}\label{tauineq} If CH holds,
\[
\tau(h_1) + \tau(h_2) - 1 \le \tau(h_1\land h_2)
\]
for any two future histories $h_1$ and $h_2$.
\end{thm}
\begin{proof}
First suppose that $h_1$ and $h_2$ are one-time histories. By Lemma \ref{PiiPii}, CH gives
\be\label{twopi}
\tau(h_1\land h_2) = \<E_0|\Pii_1\Pii_2\Pii_1|E_0\> = \<E_0|\Pii_2\Pii_1|\E_0\>.
\ee
Hence\upline
\begin{align*}
1 - \tau(h_1) - \tau(h_2) &+ \tau(h_1\land h_2) = \<E_0|(1 - \Pii_1 - \Pii_2 + \Pii_2\Pii_1)|E_0\> \\
&= \<E_0|(1-\Pii_2)(1-\Pii_1)|E_0\>\\
&= \tau(\lnot h_1\land\lnot h_2) \quad\text{ by Lemma \ref{PiiPii} again}\\
&\ge 0.
\end{align*}
Now we proceed by double induction on the lengths of $h_1$ and $h_2$: if $h_3$ is a one-time history, and the inequality holds for $h_1$ and $h_2$, then
\begin{align*}
\tau(h_1) + {} &\tau(h_2\land h_3) - 1 = \tau(h_1) + \tau(h_2) - \tau(h_2\land \lnot h_3) - 1 \quad\text{ by Theorem 8}\\
&\le \tau(h_1\land h_2) - \tau(h_2\land\lnot h_3) \;\quad\quad\quad\quad \text{ by the inductive hypothesis}\\
&= \tau(h_1\land h_2\land h_3) + \tau(h_1\land h_2\land \lnot h_3) - \tau(h_1\land \lnot h_3) \text{ by Theorem 8}\\
&\le \tau(h_1\land h_2\land h_3) \qqquad\qqquad\qqquad\qqquad\qqquad\text{ by Theorem 7}
\end{align*}
so the inequality holds for $h_1$ and $h_2\land h_3$. Hence, by induction it holds for one-time $h_1$ and any $h_2$. Now a similar induction on $h_1$ shows that it holds for all $h_1,h_2$.
\end{proof} 

\subsection{Disjunction}
\label{subsec:disjunction}

If it is clear what is meant by a future history and how to assign its truth value in quantum theory (though whether this really is clear will be discussed in the next section), it is not so clear how to approach a disjunction of histories. However, with the assumption CH, which will be a standing assumption for the remainder of this section, we can, as anticipated in section \ref{probtruth} (eq. \eqref{eq:prob}),  adopt the following definition from probability logic:
\be\label{truthdisj}
\tau(h_1\lor h_2) = \tau(h_1) + \tau(h_2) - \tau(h_1\land h_2)
\ee
since Theorems 7 and 8 assure us that this lies between 0 and 1.

We extend this to disjunctions of any finite number of future histories, i.e.\ to general elements of the lattice $\TF$, by the following definition, expressing the principle of inclusion and exclusion:
\be\label{taulor}
\tau(h_1\lor\cdots\lor h_n) = H^{(n)}_1 - H^{(n)}_2 + \cdots + (-)^{n-1}H^{(n)}_n
\ee
where
\be\label{H}
H^{(n)}_r = \sum\tau(h_{i_1}\land\cdots\land h_{i_r})
\ee
in which the sum extends over all $r$-subsets $\{i_1,\ldots,i_r\}$ of $\{1,\ldots ,n\}$.
This satisfies 
\be\label{disjinduction}
\tau(h_1\lor\cdots\lor h_n) = \tau(h_1\lor\cdots\lor h_{n-1}) + \tau(h_n) - \tau\big( (h_1\lor\cdots\lor h_{n-1})\land h_n\big)
\ee
where, by the assumed distributivity of the lattice $\TF$, the argument of the last $\tau$ can be expanded into a disjunction of $n-1$ histories. However, we do not take it as an inductive definition of $\tau(h_1\lor\cdots\lor h_n)$ since it would be necessary to show that it is well-defined, i.e.\ independent of the ordering of $h_1,\ldots,h_n$. We must show that the definition \eqref{taulor} gives a result between $0$ and $1$, and for this we need to extend Theorem \ref{hnoth}:
\begin{thm}\label{genhnoth}
If $p$ is any disjunction of histories and $h$ is a one-time history, then
\[
\tau(p) = \tau(p\land h) + \tau(p\land\lnot h).
\]
\end{thm}
\begin{proof}
This is a straightforward induction on the number of disjuncts in $p$.
\end{proof}

\begin{thm} If $\tau(h_1\lor\cdots\lor h_n)$ is given by \eqref{taulor}, then
\[
0 \le \tau(h_1\lor\cdots\lor h_n) \le 1.
\]
\end{thm}
\begin{proof}
This follows the same lines as the proof of Theorem \ref{tauineq}. If $h_n$ is a one-time history, Theorem \ref{genhnoth} gives
\[
\tau(h_1\lor\cdots\lor h_{n-1}) \ge \tau\big((h_1\lor\cdots\lor h_{n-1})\land h_n\big)
\]
and arguing by induction on the length of $h_n$, as in Theorem 7, we can extend this to general $h_n$. Hence, by \eqref{disjinduction},
\[
\tau(h_1\lor\cdots\lor h_n) \ge \tau(h_n) \ge 0.
\]

For the right-hand inequality, first suppose that $h_1,\ldots,h_n$ are one-time histories, with $\tau(h_n) = \<E_0|\Pii_n|E_0\>$; then in the definition \eqref{taulor} we have, from Lemma \ref{PiiPii},
\[
H_r^{(n)} = \sum \<E_0|\Pii_{i_1}\ldots \Pii_{i_n}|E_0\>
\]
and therefore
\begin{align*}
\tau(h_1\lor\cdots\lor h_n) &= 1 - \<E_0|(1 - \Pii_n)(1 - \Pii_{n-1})\cdots(1-\Pii_1)|E_0\>\\
&= 1 - \tau(\lnot h_1\land\cdots\land\lnot h_n)
\end{align*}.
Since $\lnot h_1\land\cdots\land\lnot h_n$ is a single history, this lies between $0$ and $1$.

The right-hand inequality can be extended to arbitrary histories $h_1,\ldots,h_n$ by a similar argument to the last part of the proof of Theorem \ref{tauineq}, using successive inductions on the lengths of $h_1,\ldots,h_n$.
\end{proof}

We can show that \eqref{truthdisj} is true in general. This is the replacement for the truth table for $\lor$ in truth-functional many-valued logic.

\begin{thm}\label{disj}
If $p$ and $q$ are any two future propositions (elements of the lattice $\TF$),
\[
\tau(p\lor q) = \tau(p) + \tau(q) - \tau(p\land q).
\]
\end{thm}
\begin{proof}
Suppose $p = h_1\lor\cdots\lor h_m$ and $q = k_1\lor\cdots\lor k_n$ where $h_n$ and $k_j$ are histories. The definition \eqref{taulor} gives expressions for $\tau(p)$, $\tau(q)$, $\tau(p\lor q)$ and $\tau(p\land q)$ which can be rewritten in a factorised form as follows.

For any set $X$, we write $2^X$ for the set of subsets of $X$, as usual. Let $U(X)$ be the commutative ring generated by the subsets of $X$ with union as the ring multiplication (in other words, $U(X)$ is a $\Z$-module with basis $2^X$, with the additional operation of union defined on the basis and extended by $\Z$-linearity). We continue to use the symbol $\cup$ for this extended operation. Then the empty set is a multiplicative identity in $U(X)$, and we will denote it by $1$.  

 Let $M = \{1,\ldots,m\}$ and $N = \{1,\ldots, n\}$. Let $M\sqcup N$ be the disjoint union of $M$ and $N$, realised as
\[
M\sqcup N = \{(1,0),\ldots,(m,0),(0,1),\ldots,(0,n)\} \subset \Z\times\Z
\]
(so $M\sqcup N$ is the union of the projections of $M\times N$ onto the axes in $\Z\times\Z$). Define real-valued functions $\tau_p: 2^M \to [0,1]$, $\tau_q: 2^N\to [0,1]$, $\tau_{p\lor q}:2^{M\sqcup N} \to [0,1]$ and $\tau_{p\land q}:2^{M\times N} \to [0,1]$ by
\begin{align*}
\tau_p(\{i_1,\ldots,i_r\}) &= \tau(h_{i_1}\lor\cdots\lor h_{i_r}),\\
\tau_q(\{j_1,\ldots,j_s\}) &= \tau(k_{j_1}\lor\cdots\lor k_{j_s}),\\
\tau_{p\lor q}\big(\{(i_1,0)\ldots,(i_r,0),(0,j_1),\ldots,(0,j_s)\}\big) &= \tau\big(h_{i_1}\lor\cdots\lor h_{i_r}\lor k_{j_1}\lor\cdots\lor k_{j_s}\big),\\
\tau_{p\land q}\big(\{(i_1,j_1),\ldots,(i_t,j_t)\}\big) &= \tau\big(h_{i_1}\lor k_{j_1}\lor\cdots\lor h_{i_t}\lor k_{j_t}\big).
\end{align*}
Then the definition \eqref{taulor} can be stated as 
\[
\tau(p) = \tau_p\big(1 - (1 - \{1\})\cup\cdots\cup(1 - \{m\})\big)
\]
with corresponding expressions for $\tau(q)$, $\tau(p\lor q)$ and $\tau(p\land q)$. To write the last two, it is convenient to introduce the abbreviations $e_n = \{(i,0)\}$, $f_j = \{(0,j)\}$  and $g_{ij} = \{(i,j)\}$ for the singleton sets in $2^{M\sqcup N}$ and $2^{M\times N}$, which generate the rings $U(M\sqcup N)$ and $U(M\times N)$. Then 
\begin{align*}
\tau(p\lor q) &= \tau_{p\lor q}\big(1 - (1 - e_1)\cup\ldots\cup(1-e_m)\cup(1-f_1)\cup\cdots\cup(i-f_n)\big),\\
\tau(p\land q) &= \tau_{p\lor q}\(1 - \bigcup_{ij}(1 - g_{ij})\).
\end{align*}

Define a map $f:2^{M\times N} \to 2^{M\sqcup N}$ by projecting $M\times N \subset \Z\times\Z$ onto the axes: 
\begin{align*}
f\big(\{(i_1,j_1),\ldots (i_r,j_r)\}\big) &= \{i_1,\ldots,i_r\}\sqcup\{j_1,\ldots,j_r\}\\
&=\{(i_1,0),\ldots,(i_r,0),(0,j_1),\ldots,(0,j_r)\}.
\end{align*}
This satisfies $f(S\cup T) = f(S)\cup f(T)$ and therefore can be extended to a ring homomorphism $f:U(M\times N) \to U(M\sqcup N)$. It is completely specified by its action on the generators: 
\be\label{gen}
f(g_{ij}) = e_n\cup f_j.
\ee

Now let $S = \{(i_1,j_1),\ldots ,(i_r,j_r)\}$ be any subset of $M\times N$. Then
\begin{align*}
\tau_{p\land q}(S)
&= \tau\(h_{i_1}\lor k_{j_1}\lor\cdots\lor h_{i_r}\lor k_{j_r}\)\\
&= \tau\(h_{i_1}\lor\cdots\lor h_{i_r}\lor k_{j_1}\lor\cdots \lor k_{j_r}\)\\
&= \tau_{p\lor q}\(e_{i_1}\cup\cdots\cup e_{i_r}\cup f_{i_1}\cup\cdots\cup f_{j_r}\)\\
&= \tau_{p\lor q}(f(S)),
\end{align*}
so $\tau_{p\land q} = \tau_{p\lor q}\circ f$. Hence
\begin{align*}
\tau(p\land q) &= \tau_{p\land q}\(1 - \bigcup_{ij}(1 - g_{ij})\)
= \tau_{p\lor q}\circ f\(1 - \bigcup_{ij}(1 - g_{ij})\)\\
&= \tau_{p\lor q}\(1 - \bigcup_{ij}\(1 - e_n\cup f_j\)\).
\end{align*}

We now note that for any $e, f_1,\ldots,f_n\in U(M\sqcup N)$ satisfying $e\cup e = e$, in particular if $e$ is one of the generators of $U(M\sqcup N)$,
\[
\bigcup_{j=1}^n (1 - e\cup f_j) = 1 - e + e\cup\bigcup_{j=1}^n(1 - f_j).
\]
which can be proved by a simple induction on $n$. Hence
\[
\bigcup_{ij}(1 - e_n\cup f_j) = \bigcup_{i=1}^m(1 - e_n + e_n\cup F)
\]
where
\upline
\[
F = \bigcup_{j=1}^n(1 - f_j).
\]
Note that $(1 - f_n)\cup(1 - f_n) = 1 - f_n$, so $F\cup F = F$. Induction on $m$ can now be used to show that 
\[
\bigcup_{i=1}^m(1 - e_n + e_n\cup F) = E + F - E\cup F
\]
where
\upline
\[
E = \bigcup_{i=1}^m (1- e_n).
\]
We now have
\begin{align}\label{taugh}
\tau(p\land q) &= \tau_{p\lor q} (1 - E - F + E\cup F)\notag\\
&= \tau_{p\lor q}(1 - E) + \tau_{p\lor q}(1 - F) - \tau_{p\lor q}(1 - E\cup F).
\end{align}
But
\upline
\begin{align*}
\tau_{p\lor q}(1 - E) &= \tau_{p\lor q}\(1 - \bigcup_{i=1}^m(1 - \{i,0\})\)\\
&= \tau_p\( 1 - \bigcup_{i=1}^m(1 - \{i\})\)\\
&= \tau(p),
\end{align*}
and similarly
\[
\tau_{p\lor q}(1 - F) = \tau(q), \qquad \tau_{p\lor q}(1 - E\cup F) = \tau(p\lor q).
\]
Eq. \eqref{taugh} therefore gives the stated result.
\end{proof}

\subsection{Negation}
\label{subsec:negation}

For a one-time history $h = F_t(\Pi)$, negation is defined by the usual orthocomplement in the lattice of closed subspaces of Hilbert space:
\[
\lnot h = F_t(1 - \Pi)
\]
with truth value
\be\label{nottruth}
\tau(\lnot h) = 1 - \tau(h).
\ee
For a general history $h = h_1\land\cdots\land h_r$ where $h_n$ are one-time histories, we take negation to be given by
\[
\lnot h = \lnot h_1\lor\cdots\lor\lnot h_r,
\]
and for a general proposition $p = k_1\lor\cdots\lor k_r$, where $k_n$ are multi-time histories,
\[
\lnot p = \lnot k_1\land\cdots\land\lnot k_r.
\]
The truth values of such general $h$ and $p$ are aleady covered by the definition in the previous section. Our aim in this subsection is to prove that \eqref{nottruth} still holds.

\begin{lemma}\label{hnothagain} If $h_1$ and $h_2$ are any two histories,
\[
\tau(h_1) = \tau(h_1\land h_2) + \tau(h_1\land\lnot h_2)
\]
\end{lemma}
\begin{proof}
By induction on the length of $h_2$. By Theorem \ref{hnoth}, the theorem holds for all $h_1$ and all one-time histories $h_2$. Suppose it holds for all $h_2$ of length $n$, and let $h'_2$ be a history of length $n+1$, so $h'_2 = h_2 \land h_3$ where $h_2$ has length $n$ and $h_3$ has length $1$. Then
\begin{align*}
\tau(h_1 \land h'_2) + \tau(h_1 \land \lnot h'_2)
&= \tau(h_1 \land h_2 \land h_3)  + \tau\big(h_1 \land (\lnot h_2 \lor\lnot h_3)\big)\\
&= \tau(h_1 \land h_2 \land h_3)  +\tau\big((h_1\land\lnot h_2)\lor(h_1\land\lnot h_3)\big)\\
&= \tau(h_1 \land h_2 \land h_3)  +\tau(h_1\land \lnot h_2) + \tau(h_1\land\lnot h_3)\\ 
&\qqquad -\tau(h_1\land\lnot h_2\land\lnot h_3)\\
&= \tau(h_1 \land h_2 \land h_3)  + \tau(h_1\land\lnot h_2\land h_3) + \tau(h_1\land \lnot h_3)\\ &\qqquad\text{by Theorem \ref{hnoth}}\\
&= \tau(h_1\land h_3) + \tau(h_1\land\lnot h_3)\\
&\qqquad \text{by the inductive hypothesis}\\
&= \tau (h_1) \quad\text{by Theorem \ref{hnoth} again}.
\end{align*}
\end{proof}

\begin{lemma}
If h is any history,
\[
\tau(\lnot h) = 1 - \tau(h)
\]
\end{lemma}
\begin{proof}
By induction on the length of $h$. Suppose the result holds for all histories of length $n$, and let $h'$ be a history of length $n+1$, so that $h' = h_0\land h$ where $h_0$ is a one-time history and $h$ has length $n$. Then
\begin{align*}
\tau(\lnot h) &= \tau(\lnot h_0\lor\lnot h')\\ 
&= \tau(\lnot h_0) + \tau(\lnot h) - \tau(\lnot h_0 \land\lnot h).
\end{align*}
But $\tau(\lnot h) = 1 - \tau(h)$ by the inductive hypothesis, and
\[
\tau(\lnot h_0) - \tau(\lnot h_0\land\lnot h) = \tau(\lnot h_0\land h) \text{ by Theorem \ref{hnothagain},}
\]
so\upline
\begin{align*}
\tau(\lnot h') &= 1 - \tau(h) + \tau(\lnot h_0\land h)\\
&= 1 - \tau(h_0\land h) \quad\text{ by Theorem \ref{hnoth}}\\
&= 1 - \tau(h').
\end{align*}
The result holds for $n = 1$ by \eqref{nottruth}, and therefore holds for all $n$.
\end{proof}

\begin{thm}
If $p$ is any disjunction of histories,
\[
\tau(\lnot p) = 1- \tau(p)
\]
\end{thm}
\begin{proof}
By induction on the number of disjoined histories in $p$. Suppose the theorem holds for all $p = h_1\lor\cdots\lor h_n$, and let $p' = p\lor k$ where $k$ is another history. Then
\begin{align*}
\tau(\lnot p') &= \tau(\lnot p \land \lnot k)\\
&= \tau(\lnot p) + \tau(\lnot k) - \tau(\lnot p\lor\lnot k)\\
&= \tau(\lnot p) + \tau(\lnot k) - \tau\big(\lnot(p\land k)\big)\\
&= 1 - \tau(p) + 1 - \tau(k) - [1 - \tau(p\land k)]
\end{align*}
by the inductive hypothesis, since $p\land k$ is also a conjunction of $n$ histories $h_n\land k$. Hence
\[
\tau(\lnot p') = 1 - \tau(p\lor k) = 1 - \tau(p')
\]
and the theorem is established for all $p$ by induction.
\end{proof}

A similar double induction can be used to prove the general version of Theorem \ref{hnoth}:
\begin{thm}\label{pnotq}
If $p$ and $q$ are any two elements of the lattice $\TF$,
\[
\tau(p) = \tau(p\land q) + \tau(p\land\lnot q).
\]
\end{thm}

We can now state the general logical properties of conjunction and disjunction in this temporal logic:

\begin{thm} \label{truth}
Let $p$ and $q$ be any two temporal propositions. 

\semilinespace

\emph{(i)}\upline
\[
\tau(p\land q) = 1 \; \iff \; \tau(p) = \tau(q) = 1;
\]

\emph{(ii)}\upline
\[
\tau(p) = 0 \text{ or } \tau(q) = 0 \;\impl\; \tau(p\land q) = 0;
\]

\emph{(iii)}\upline
\[
\tau(p\lor q) = 0 \;\iff\; \tau(p) = \tau(q) = 0;
\]

\emph{(iv)}\upline
\[
\tau(p) = 1 \text{ or } \tau(q) = 1 \;\impl\; \tau(p\lor q) = 1.
\]
\end{thm}
\begin{proof}
\noindent (i) If $\tau(p\land q) = 1$, then, by Theorem \ref{pnotq},
\[
\tau(p) = \tau(p\land q) + \tau(p\land\lnot q) \ge 1
\]
and therefore $\tau(p) = 1$. Conversely, if $\tau(p) = \tau(q) = 1$, then 
\[
\tau(p\land q) = \tau(p) + \tau(q) - \tau(p\lor q) \ge 1
\]
and therefore $\tau(p\land q) = 1$.

\semilinespace

\noindent (ii) If $\tau(p) = 0$, then
\begin{align*}
\tau(p\land q) &= \tau(p) + \tau(q) - \tau(p\lor q)\\
&= \tau\big(\lnot(p\lor q)\big) - \tau(\lnot q)\\
&= \tau(\lnot p \land \lnot q) - \tau(\lnot q)\\
&= -\tau(\lnot q\land p) \qquad \text{by Theorem \ref{pnotq}}\\
&\le 0,
\end{align*}
so $\tau(p\land q) = 0$.

\semilinespace

(iii) and (iv) follow from these by writing $\tau(p\lor q) = 1 - \tau(\lnot p\land \lnot q)$. 
\end{proof}

A possibly disturbing feature of this list of properties is the one-sidedness of the implications (ii) and (iv): $p\land q$ might be false without either $p$ or $q$ being false (though (i) shows that if $p\land q$ is not definitely true, then $p$ and $q$ cannot both be definitely true); and the truth of $p\lor q$ does not imply the truth of either $p$ or $q$ (though (iii) shows that this implication does hold if ``truth" ($\tau = 1$) is replaced by ``possible truth" ($\tau \neq 0$)). The first of these has already been discussed, but the second, concerning truth rather than falsity, might be felt to cast doubt on whether $\lor$ can legitimately be regarded as a generalised form of ``or" (a similar objection has been made to quantum logic \cite{QMPN}). Such an objection would be an argument against any identification of probabilities with truth values. However, far from being objectionable, this feature of disjunction seems to be necessary in a temporal logic that can deal with indeterminism. As Prior noted (\cite{Prior:logic} p.244), Aristotle's assertion ``Either there will be a sea-battle tomorrow or there won't" should be taken to mean, not what it appears to mean, but `` `Either ``There is a sea-battle going on" or ``there is no sea-battle going on" ' will be true tomorrow"; in symbols,
\[
\text{By } F_t(p) \lor F_t(\lnot p) \text{ Aristotle meant } F_t(p\lor \lnot p).
\]
But in the logic proposed here $F_t$ is taken to be a homorphism, so that
\[
F_t(p) \lor F_t(q) = F_t(p\lor q),
\]
showing that the meaning of the connective $\lor$ in this system is in accordance with Prior's reading of Aristotle. Aristotle's usage (or confusion, if that is what it is) is of course very common. I take Prior's elucidation of it as grounds for claiming that the disjunction $\lor$ occurring in this temporal logic is in fact the ``or" of common usage in talk about the open future.

\section{Retrospect}

In this final section I will consider to what extent the proposals in this paper meet the objectives set out in Section 2.

\paragraph{The open future} The reader may well have observed that there is nothing new or specifically quantum-mechanical about the indeterministic world of experiences which emerges from quantum mechanics, as described in Section 2; this is the classical picture of an open future. The assumption CH in Section 4.2 simply ensures the classical nature of our logic. Thus the logic developed here is a temporal logic appropriate to a metaphysics of time that incorporates an open future. It can be summarised as follows:

{\bf 1.} The lattice $\T$ of tensed (and dated) propositions is formed from a lattice $\E$ of tenseless propositions by means of lattice homomorphisms $\{P_t, N, F_t: t\in \R, t>0\}$ whose images generate $\T$.

{\bf 2.} A set of truth values is a map $\tau:\E\to\R$ satisfying

\semilinespace

\noindent (i) $0\le \tau(p) \le 1$;

\semilinespace

\noindent (ii) $\tau(p) = 0$ or $1$ if $p$ belongs to $N(\E)$ or $P_t(\E)$;

\semilinespace

\noindent (iii) $\tau(p\land q) + \tau(p\lor q) = \tau(p) + \tau(q)$;

\semilinespace

\noindent (iv) $\tau(p\land q) \le \tau(p), \hspace{1em} \tau(q) \le \tau(p\lor q)$.

\semilinespace
 
The logic based on these axioms will be discussed elsewhere.

\paragraph{Strictly Quantum}  The reader may also have observed that I have not kept to the faith boldly proclaimed in Section 2. I there renounced the devil and all his works, such as the collapse postulate, and his pomps, such as a separate classical realm. But the formula I took for the truth value (or probability) of a history is taken from a calculation of probability based on the collapse postulate; and the Consistent Histories assumption, which I introduced in order to   
get a recognisable logic, is just a way of eliminating the characteristic features of quantum mechanics. This assumption can be justified by decoherence theory within quantum mechanics, but only if one restricts the set of possible histories (for example, it will not be true if one wants to include times that are arbitrarily close to each other). Moreover, there is no justification for the definite past propositions $P_t(p)$; there is no thin red line between experience states going into the past, any more than there is one going into the future. There are only statements taken from present memory, which are aspects of the present experience states (in other words, statements about the past can only be in the perfect tense -- which is one of the three aspects of the present tense in English \cite{Joos}).

These objections can be met by restricting the logic as follows. The present sublattice $N(\E)$ describes present experience, including memory (so statements about the past occur in the perfect tense as elements of $N(\E)$). From the context of a particular experience state $\eta_0$, these elements of $N(\E)$ have truth values restricted to $\{0,1\}$. It is possible to make future-tense statements only about experiences at one future time. A statement about a future history must be regarded as an element of $F_{t_\text{f}}(\E)$, where $t_\text{f}$ is the last time in the history; the other elements of the history are in the future perfect tense, referring to the contents of memory at time $t_\text{f}$ (this, after all, is the only way we can verify a statement referring to a number of future times: we wait until the last time referred to and then consult records of earlier times). 

In this austere approach, the set of propositions is not a lattice but is the union of the sublattices $N(\E)$ and $F_t(\E)$ for all positive real numbers $t$. A conjunction or disjunction of elements of different sublattices is not a well-formed proposition in this logic. However, it can be regarded as a statement in the metalanguage, with the usual bivalent truth values: if $e$ and $f$ are experience statements, so that $F_s(e)$ and $F_t(f)$ are statements in the future tense, then the disjunction $F_s(e)\lor F_t(f)$ means, as usual, that one of the two propositions is (definitely) true, i.e.\ ``either $\tau(F_s(e)) = 1$ or $\tau(F_t(f)) =1$".

 \section{Conclusion}
 
 In this paper I have proposed a way to understand the statements about the future made by a sentient physical system in a universe described by quantum theory. I take it that we are such sentient physical systems, living in such a universe. The proposal is that the statements any one of us can make, from his or her own perspective in the universe, form a lattice with the usual logical operations of conjunction, disjunction and negation obeying conventional laws, but that the appropriate truth values for such statements (or ``histories") are many-valued, with values in the unit interval $[0,1]$ of real numbers, and are to be identified with probabilities. The paper has explored the logical properties of the truth values given by quantum theory, delineating which of these properties require the special assumption of ``consistent histories". We find that this assumption is sufficient to prove all the logical properties to be expected from the identification of truth values with probabilities.


\end{document}